\documentclass[12pt]{colt2015} 

\usepackage{amsfonts}
\usepackage{algorithm}
\usepackage[noend]{algorithmic}
\usepackage{color}
\usepackage{bbm}

\title[Recursive Eigen-Decomposition]{Max vs Min: 
  Tensor Decomposition and ICA\\
with nearly Linear Sample Complexity}
\usepackage{times}

 \coltauthor{\Name{Santosh S. Vempala} \Email{vempala@gatech.edu}\\
 \addr School of Computer Science, Georgia Tech
 \AND
 \Name{Ying Xiao} \Email{ying.xiao@gatech.edu}\\
 \addr School of Computer Science, Georgia Tech
 }

\begin{document}

\maketitle


\newtheorem{claim}[theorem]{Claim}
\newtheorem{fact}[theorem]{Fact}

\newcommand{\inshort}[1]{}
\newcommand{\infull}[1]{#1}
\def\reals{\mathbb{R}}
\def\R{\reals}
\def\N{\mathbb{N}}
\def\sph{\mathbb{S}}
\def\psd{\succcurlyeq}
\def\eps{\epsilon}
\def\diam{\mathrm{D_{K}}}
\def\Lag{\mathcal{L}}
\def\poly{\mathrm{poly}}
\def\HH{{\cal H}}
\def\B{\mathbb{B}}
\def\C{\mathbb{C}}
\def\Z{\mathbb{Z}}
\def\geqs{\succcurlyeq}
\def\leqs{\preccurlyeq}
\def\argmax{\mathrm{argmax}}
\def\cum{\mathrm{cum}}
\def\Var{\mathrm{Var}}

\newcommand{\ball}[2]{#2\mathbb{B}_{#1}}
\newcommand{\abs}[1]{\left|#1\right|}
\newcommand{\norm}[1]{\left\|#1\right\|}
\newcommand{\sqnorm}[1]{\left\|#1\right\|^{2}}
\newcommand{\size}[1]{\left|#1\right|}
\newcommand{\prob}[1]{{\sf Pr}\left(#1\right)}
\newcommand{\probsub}[2]{{\sf Pr}_{#1}\left(#2\right)} 
\newcommand{\tr}[1]{\mathrm{tr} \left(#1\right)}
\newcommand{\eigmin}[1]{\lambda_{\min}(#1)}
\newcommand{\eigmax}[1]{\lambda_{\max}(#1)}
\newcommand{\vol}[1]{\operatorname{vol}\left(#1\right)}
\newcommand{\E}[1]{\mathbb{E}\left(#1\right)}
\newcommand{\T}[1]{#1^\mathrm{T}}
\newcommand{\ip}[2]{\left<{ #1},{#2}\right>}
\newcommand{\pard}[2]{\frac{\partial #1 }{\partial #2 }}
\newcommand{\spn}[1]{\mathrm{span}\left(#1\right)}
\newcommand{\EE}[2]{\mathbb{E}_{#1}\left(#2\right)}
\newcommand{\ind}[1]{\mathbb{I}_{#1}}
\newcommand{\sgn}[1]{\mathrm{sgn}\left( #1 \right)}
\newcommand{\conv}[1]{\mathrm{conv}\left\{ #1 \right\}}
\newcommand{\diag}[1]{\mathrm{diag}\left( #1 \right)}
\newcommand{\mat}[1]{\mathrm{mat}\left( #1 \right)}
\newcommand{\rank}[1]{\mathrm{rank}\left( #1 \right)}
\renewcommand{\det}[1]{\mathrm{det}\left( #1 \right)}
\renewcommand{\vec}[1]{\mathrm{vec}\left( #1 \right)}
\newcommand{\re}[1]{\operatorname{Re}\left( #1 \right)}
\newcommand{\im}[1]{\operatorname{Im}\left( #1 \right)}
\newcommand{\colspan}[1]{\operatorname{\mathsf{colspan}}\left( #1 \right)}
\newcommand{\maxgap}[1]{\operatorname{maxgap}\left( #1 \right)}
\newcommand{\bE}[1]{\bar{\mathbb{E}}\left(#1\right)}
\newcommand{\bMm}{\bar{M}_{\mu}}
\newcommand{\bMl}{\bar{M}_{\lambda}}
\newcommand{\bP}{\bar{P}}

\newcommand{\placeholder}{K}

\newcommand{\Vr}[1]{\mathrm{Var}\left(#1\right)}
\newcommand{\Conjug}[1]{}
\newcommand{\suchthat}{\;\ifnum\currentgrouptype=16 \middle\fi|\;}

\newcommand{\tocheck}[1]{\textcolor{blue}{#1}}
\newcommand{\angles}[1]{\left\langle #1 \right\rangle}

\newcommand{\expn}[1]{\mathrm{exp}\left(#1\right)}
\newcommand{\mytilde}{\raise.17ex\hbox{$\scriptstyle\mathtt{\sim}$}}
\newcommand{\nE}{\mathbb{E}}
\newcommand{\M}{{\bf M}}
\newcommand{\hide}[1]{}

\begin{abstract}
  We present a simple, general technique for reducing the sample 
  complexity of matrix and tensor decomposition algorithms applied to distributions. We use the technique to
  give a polynomial-time algorithm for standard ICA with   
  sample complexity nearly {\em linear} in the dimension, thereby improving substantially on previous bounds. 
The analysis is based
  on properties of random polynomials, namely the spacings of an
  ensemble of polynomials. Our technique also applies to other applications of tensor decompositions, including spherical Gaussian mixture models.
\end{abstract}




\section{Introduction}
Matrix and tensor decompositions have proven to be a powerful
theoretical tool in a number of models of unsupervised learning, e.g., Gaussian
mixture models \citep{vempala2004spectral, AndersonGMM,HsuK13,GVX14}, Independent Component
Analysis \citep{FriezeJK96,NguyenR09,BelkinRV12, AroraGMS12,GVX14}, and topic
models \citep{AnandkumarTensorDecomp} to name a few recent results. 
A common obstacle for these approaches is that while the tensor-based algorithms are polynomial in the dimension,
they tend to have rather high sample complexity. As data abounds, and labels are expensive, efficient
methods for unsupervised learning --- recovering latent structure from unlabeled data --- are becoming more important.

Here we consider the classic problem of Independent Component Analysis
(ICA), which originated in signal processing and has become a
fundamental problem in machine learning and statistics, finding
applications in diverse areas, including neuroscience, computer vision
and telecommunications \citep{hastie2009elements}. More recently, it has been used as a tool for
sparsifying layers in deep neural nets \citep{ngiam2010tiled}. The input to the problem is a set of
i.i.d. vectors from a distribution in $\R^n$. The latter is assumed to
be an unknown linear transformation of an unknown distribution with
independent $1$-dimensional component distributions. More precisely,
each observation $x \in \R^n$ can be written as $x = As$, where $A \in
\R^{n \times n}$ is an unknown matrix and $s \in \R^n$ has components
$s_1, \ldots, s_n \in \R$ generated independently (from possibly
different one-dimensional distributions). ICA is the problem of
estimating the matrix $A$, the basis of the latent product
distribution, up to a scaling of each column and a desired error
$\eps$. One cannot hope to recover $A$ if more than one $s_i$
is Gaussian --- any set of orthogonal directions in the subspace
spanned by Gaussian components would also be consistent with the
model. Hence the model also assumes that at most one component is
Gaussian and the other component distributions differ from being
Gaussian in some fashion; the most common assumption is that the
fourth cumulant, also called the {\em kurtosis}, is nonzero (it is zero for a
Gaussian).

Our main result is a polynomial-time algorithm for ICA, under the fourth cumulant assumption, using only {\em $\tilde{O}(n)$} samples, which is nearly optimal. This improves substantially on previous polynomial-time algorithms \citep{FriezeJK96, NguyenR09, AroraGMS12,  BelkinRV12, AnandkumarTensorDecomp, GVX14}, which all require a higher polynomial number of samples  ($\Omega(n^5)$ or higher). Our technique also applies more broadly to tensor decomposition problems where the tensor is estimated from samples. In particular, it applies to tensor decomposition methods used in analyzing topic models and learning mixtures of spherical Gaussians with linearly independent means; it extends to the setting where the model is corrupted with Gaussian noise. Before stating our results precisely, we place it in the context of related work.

\subsection{Related work}
Tensor decomposition is a major technique in the literature on ICA. The latter is vast, with many proposed algorithms (see \cite{ComonJutten, ICAbook}).  \cite{FriezeJK96} were the first to provide rigorous finite
sample guarantees, with several recent papers improving their
guarantee for the fully determined case when $A$ is nonsingular
\citep{NguyenR09, AroraGMS12, BelkinRV12,
  AnandkumarTensorDecomp}. These results either assume that component
distributions are specific or that the fourth moment is bounded away
from that of a Gaussian. \cite{GVX14} recently gave an algorithm that can deal with
differences from being Gaussian at any moment. The algorithm, called Fourier PCA, can handle unknown Gaussian noise. It extends to the  underdetermined setting where the signal $s$ has more components than the observation $x$ (so $A$ is rectangular), resulting in a polynomial-time algorithm under a much milder condition than the nonsingularity of $A$.

The main technique is all these papers can be viewed as efficient tensor decomposition. For a 
Hermitian matrix $A \in \R^{n \times n}$, one can give an
orthogonal decomposition into rank $1$ components:
\begin{align*}
  A = \sum_{i=1}^n \lambda_i v_i v_i^T.
\end{align*}
This decomposition, especially when applied to covariance matrices, is
a powerful tool in machine learning and theoretical computer science. 
The generalization of this to tensors is not
straightforward, and many versions of this
decomposition lead directly to NP-hard problems \citep{hillar2009most, Brubaker09}. The
application of tensor decomposition to ICA was proposed by
\cite{Cardoso89}. Such decompositions were used by \cite{anandkumar2012spectral} and \cite{HsuK13}
to give provable algorithms for various latent variable models. \cite{GVX14} extended these decompositions to a more general setting where
the rank-one factors need not be linearly independent (and thus might
be many more than the dimension).

In spite of these polynomial algorithms, even for standard ICA with a
square matrix $A$, the dependence of the time and sample complexity on
the conditioning of $A$ and the dimension $n$ make them impractical
even in moderately high dimension. 

\subsection{Our technique and results}

Suppose we are interested in finding the eigenvectors of a matrix estimated from samples, e.g., a covariance matrix of a distribution in $\R^n$. Then, for the decomposition to be unique, we need the eigenvalues to be distinct, and the sample complexity grows with the inverse square of the {\em minimum} eigenvalue gap. In many situations however, for such matrices, the {\em maximum} eigenvalue gap is much larger. The core idea of this paper is to use as many samples as needed to estimate the largest gap in the eigenvalues accurately, find the corresponding subspaces (which will be stable due the large gap), then recurse in the two subspaces by projecting samples to them. Our target application will be to ICA, but the technique can be used for tensor decomposition in other settings.

Our main result is a polynomial-time algorithm for ICA using only a nearly linear number of samples. Since each column of $A$ can only be recovered up to a scaling of the column, we can assume w.l.o.g. that $s$ is isotropic.
\begin{theorem}\label{thm:main}
  Let $x \in \R^n$ be given by an ICA model $x=As$, where $A \in \R^{ n
    \times n}$, the components of $s$ are independent, $\norm{s} \le K
  \sqrt{n}$ almost surely, and for each $i$, $\E{s_i}=0, \E{s_i^2}=1$, and  $\abs{\cum_{4}(s_i)}=\abs{\E{|s_i|^4}-3}
  \ge \Delta$. Let $\M \ge \max_i \E{|s_i|^5}$. Then for any $\eps < \Delta^3/(10^8 \M^2 \log^3 n)$, 
  with $\sigma = \Delta/(1000\M^2 \log^{3/2}n)$, with high probability, \textbf{Recursive FPCA}
 finds vectors $\{b_1, \ldots, b_n\}$
  such that there exist signs $\xi_i = \pm 1$ satisfying 
$\norm{A^{(i)} - \xi_i b_i} \le \eps\|A\|_2$
  for each column $A^{(i)}$ of $A$, using 
\[
O\left(n \cdot \frac{K^2 \M^4 \log^7 n}{\Delta^6\eps^2}\right) = O^*(n)
\]
samples. The running time is bounded by the time to compute $\tilde{O}(n)$ Singular Value Decompositions on real symmetric matrices of size $n \times n$.
\end{theorem}


This improves substantially on the previous best sample complexity.
The algorithm is a {\em recursive} variant of standard tensor decomposition. 
It proceeds by first puting the distribution in near-isotropic position, then computing the
eigenvectors of a reweighted covariance matrix; these eigenvectors are essentially the columns $A^{(i)}$. To do so
accurately with few samples requires that the eigenvalues of the random
matrix are well-spaced.  To
estimate all the eigenvectors, we need large spacings between all
$n-1$ adjacent eigenvalue pairs, i.e., we need that $\min_i
\lambda_{i+1} - \lambda_i$ should be large, and the complexity of this
method is polynomial in the inverse of the minimum gap.

The idea behind our new algorithm is very simple: instead of estimating all
all the gaps (and eigenvectors) accurately in one shot, we group the eigenvectors
according to which side of the {\em largest} gap $\max_i \lambda_{i+1} -
\lambda_i$ they fall. The vectors
\emph{inside} either of these subspaces are not necessarily close to
the desired $A^{(i)}$, but we can proceed recursively in each
subspace, {\em re-using the initial sample}. The key fact though, is that at each stage, we only need 
the maximum gap to be large and thus the number of samples needed is much
smaller. As a motivating example, if we pick $n$ random points from
$N(0,1)$, the minimum gap is about $O(1/n^2)$ while the maximum gap in
expectation is $\Omega(1/\sqrt{\log n})$. Since the sample complexity
grows as the square of the inverse of this gap, this simple idea
results in a huge saving. For the tensor approach to ICA, the complexity goes down to $\tilde{O}(n^2)$. 
To go all the way to linear, we will apply the maxgap idea and recursion to the Fourier PCA algorithm of 
\cite{GVX14}.

This paper has three components of possibly general interest: first, an ICA algorithm with nearly linear and thus
nearly optimal sample complexity; second, the use and analysis of maximum spacings of the
eigenvalues of random matrices as a tool for the design and
analysis of algorithms (typically, one tries to control the minimum,
and hence, all the gaps); and finally, our proof of the maximum
spacing uses a simple coupling technique that
allows for decoupling of rather complicated dependent processes. We note that our
algorithmic result can be applied to learning mixtures of spherical
Gaussians with linearly independent means and to ICA with Gaussian
noise where $x = As+ \eta$ and $\eta$ is from an unknown (not
necessarily spherical) Gaussian distribution. We do not treat these
extensions in detail here as they are similar to \cite{GVX14}, but with 
the improved sample complexity of the core algorithm. Our approach can also be used to improve the sample complexity of other applications of tensor methods
\citep{AnandkumarTensorDecomp}, including learning hidden Markov models and latent topic models.

\section{Outline of approach}

Given a tensor $T = \sum_{i=1}^n \alpha_i v_i \otimes v_i \otimes v_i \otimes v_i$, we consider the matrix
\[
T(u,u) = \sum_{i=1}^n \alpha_i (v_i \cdot u)^2 v_i \otimes v_i = V\diag{\alpha_i(v_i \cdot u)^2}V^T.
\] 
If $u$ is chosen randomly, then with high probability, the diagonal entries will be distinct. For a Gaussain $u$, the diagonal entries will themselves be independent Gaussians, and the minimum gap will be $O(1/n)$. However, the {\em maximum} gap will be much larger 
(inverse logarithmic), and we exploit this in the following algorithm. The same approach applies to tensors of order $k$, where we have $T(u,u,\ldots,u)$ with $k-2$ arguments and the RHS coefficients are 
$\alpha_i (v_i \cdot u)^{k-2}$. We describe the algorithm below for fourth order tensors ($k=4$), but an almost identical algorithm 
works for any $k \ge 3$, with $T(u,u)$ replaced by $T(u,\ldots,u)$ with $k-2$ copies of $u$ as arguments.
\begin{figure}[h!]
\begin{center}
\fbox{\parbox{\textwidth}{
\begin{minipage}{6in}
\vspace{0.1in}
{\bf Recursive-Decompose}($T$,  Projection matrix $P \in \R^{n
  \times \ell}$)
\begin{enumerate}
\item (Gaussian weighting) Pick a random vector $u$ from $N(0,  1)^n$.

\item (SVD) Compute the spectral decomposition of $P^T T(u,u)P$ to obtain $\{ \lambda_i \}$, $\{v_i\}$.

\item (Eigenvalue gap) Find the largest gap $\lambda_{i+1} -
  \lambda_i$. If the gap is too small, pick $u$ again. Partition the eigenvectors into $V_1 = \{v_1, \ldots, v_i \}$
  and $V_2 = \{ v_{i+1}, \ldots, v_\ell \}$.

\item (Recurse) For $j=1,2$: if $|V_j|=1$ set $W_j = V_j$, else $W_j=\textrm{Recursive-Decompose}(T, P V_1)$.

\item Return $[ W_1 \quad W_2]$.
\end{enumerate}
\end{minipage}
}}
\end{center}
\end{figure}
To apply the algorithm for a specific unsupervised learning model,  we just have to define the tensor $T$ appropriately \citep{AnandkumarTensorDecomp}.  For ICA, this is 
\[
T = \E{x \otimes x \otimes x \otimes x} - M
\]
where $M_{ijkl} = \E{x_i x_j}\E{x_kx_l} + \E{x_ix_k}\E{x_jx_l} + \E{x_i x_l}\E{x_j x_k}$.
As we show in Theorem \ref{thm:maxgap}, the maximum gap grows as $\Omega(1)$, while the 
minimum gap is $O(1/n^2)$. While this already improves the known sample complexity bounds to quadratic in the dimension, it is still $\Omega(n^2)$.

\subsection{Fourier PCA}
To achieve near-linear sample complexity for ICA, we will apply the recursive decomposition idea to the Fourier PCA approach of 
\cite{GVX14}. For a
random vector $x \in \R^n$ distributed according to $f$, the
characteristic function is given by the Fourier transform
\begin{align*}
  \phi(u)  = \E{ e^{iu^Tx}} = \int f(x) e^{iu^Tx} dx.
\end{align*}
In our case, $x$ will be the observed data in the ICA problem.
The \emph{second characteristic function} or
\emph{cumulant generating function} given by $\psi(u) = \log(
\phi(u))$. For $x = As$,
we define the component-wise characteristic functions with
respect to the underlying signal variables 
\begin{equation}\label{psi}
\phi_j(u_j) = \E{ e^{i u_j
    s_j}} \quad \mbox{ and } \quad \psi_i(u_j) = \log( \phi_j(u_j)) = \sum_{k=1}^\infty \cum_k(s_j)\frac{(iu_j)^k}{k!}.
\end{equation}
Here $\cum_k(y)$ is the $k$'th cumulant of the random variable $y$, a polynomial in its first $k$ moments (the second characteristic function is thus also called the cumulant generating function).
Note that both
these functions are with respect to the underlying random variables
$s_i$ and not the observed random variables $x_i$. For convenience, we
write $g_i = \psi_i''$.

The reweighted covariance matrix in the algorithm is
precisely the Hessian $D^2 \psi$:
\begin{align*}
D^2 \psi &= -\frac{\E{ (x- \mu_u)(x- \mu_u)^T
      e^{iu^Tx}}}{\E{ e^{iu^Tx}}} = \Sigma_u,
\end{align*}
where $\mu_u = \E{ x e^{iu^Tx}} / \E{ e^{iu^Tx}}$. This matrix $D^2
\psi$ has a very special form; suppose that $A=I_n$:
\begin{align*}
  \psi(u)  = \log \left( \E{ e^{iu^Ts}} \right)
   = \log \left( \E{ \prod_{j=1}^n e^{iu_js_j}}\right) 
   = \sum_{j=1}^n \log( \E{ e^{iu_js_j}} )
   = \sum_{j=1}^n \psi_j(u_j).
\end{align*}
Taking a derivative will leave only a single term
$ \frac{ \partial \psi}{\partial u_j} = \psi'_j(u_j)$.
And taking a second derivative will leave only the diagonal terms
\begin{align*}
  D^2 \psi = \diag{ \psi_j''(u_j)} = \diag{ g_j( u_j)}.
\end{align*}
Thus, diagonalizing this matrix will give us the columns of $A = I_n$,
provided that the eigenvalues of $D^2 \psi$ are nondegenerate. The general case for $A \neq I_n$ follows from the chain rule. The matrix $D^2 \psi$ is symmetric (with complex eigenvalues), but
not Hermitian; it has the following decomposition as observed by \cite{DBLP:journals/sigpro/Yeredor00}. The statement below holds for any nonsingular matrix $A$, we use it for unitary $A$,  since we can first place $x$ in isotropic
position so that $A$ will be effectively unitary. 

\begin{lemma}\label{lemma:firstderivative}
  Let $ x \in \R^n$ be given by an ICA model $x = As$ where $A \in
  \R^{n \times n}$ is nonsingular and $s\in \R^n$ is an
  independent random vector. Then
  \begin{align*}
    D^2 \psi =  A \diag{ g_i( (A^Tu))_i} A^T.
  \end{align*}
\end{lemma}

To obtain a robust algorithm, we need the eigenvalues of $D^2 \psi$
being adequately spaced (so that the error arising from sampling does
not mix the columns of $A$). For this, \cite{GVX14} pick a random vector $u \sim N(0, \sigma^2 I_n)$, so that the
$g_i(u_i)$ are sufficiently anti-concentrated with $\sigma$ small enough and
the number of samples large enough so that with high probability for all
pairs $i\neq j$, they could guarantee 
$\abs{ g_i( (A^Tu)_i) - g_j( (A^Tu)_j)} \ge \delta$
for a suitable $\delta$, leading to a (large) polynomial complexity.

\section{Recursive Fourier PCA}
We  partition the eigenvectors into two sets according to where their
eigenvalues fall relative to the maximum gap,  
project to the two subspaces spanned by these sets and
recurse, {\em re-using the initial sample}. The parameter $\sigma$ below can be safely set according to Theorem \ref{thm:main} but  in practice we suggest starting with a constant $\sigma$ and halving it till the output of the algorithm has low error (which can be checked against a new sample).
\begin{figure}[h!]
\begin{center}
\fbox{\parbox{\textwidth}{
\begin{minipage}{6in}
\vspace{0.1in}
{\bf Recursive FPCA}($\sigma$, Projection matrix $P \in \R^{n
  \times k}$)
\begin{enumerate}


\item (Isotropy) Find an isotropic transformation $B^{-1}$ with
\[
B^2 = \frac{1}{|S|}\sum_{x \in S}P^T(x-\bar{x})(x-\bar{x})^TP.
\]
\item (Fourier weights) Pick a random vector $u$ from $N(0, \sigma^2
  I_n)$. For every $x$ in the sample $S$, compute $y = B^{-1}P^Tx$, and
  its Fourier weight
\[
w(y) = \frac{e^{iu^T x }}{\sum_{x \in S} e^{iu^T x }}.
\] 
\item (Reweighted Covariance) Compute the covariance matrix of the points $y$ reweighted by $w(y)$
\[
\mu_u = \frac{1}{\abs{S}}\sum_{y \in S} w(y) y \quad \mbox{ and }
\quad \Sigma_u = -\frac{1}{\abs{S}}\sum_{y \in
  S}w(y)(y-\mu_u)(y-\mu_u)^T.
\]
\item (SVD) Compute the spectral decomposition $\{ \lambda_i \}$, $\{v_i\}$
  of $\textrm{Re}(\Sigma_u)$.

\item (Eigenvalue gap) Find the largest gap $\lambda_{i+1} -
  \lambda_i$. If the gap is too small, pick $u$ again. Partition the eigenvectors into $V_1 = \{v_1, \ldots, v_i \}$
  and $V_2 = \{ v_{i+1}, \ldots, v_k \}$.

\item (Recursion) For $j=1,2$: if $|V_j| = 1$ set $W_j = V_j$, else $W_j=\textrm{Recursive FPCA}(\sigma, P
  V_j)$.

\item Return $[ W_1 \quad W_2]$.
\end{enumerate}
\end{minipage}
}}
\end{center}
\end{figure}
The recursive decomposition step can be carried out simultaneously for all the decomposed blocks. In other words, viewing the decomposition as a tree, the next step can be carried for all nodes at the same level, with a single SVD, and a single vector $u \sim N(0,1)^n$. This is the same as doing each block separately with the vector $u$ 
projected to the span of the block. 

\section{Analysis}

The analysis of the recursive algorithm has three parts. We will show that: 
\begin{enumerate}
\item There is a large gap in the set if diagonal values, i.e.,
the set $\{g_i((A^Tu)_i)\}$. Since we make the distribution isotropic, $A^Tu$ has the same distribution as $u$, so we can focus on 
$g_i$ evaluated at independent Gaussians.
\item There is a partition of the columns of $A$ into two subsets whose spans are $V$ and $\bar{V}$, so that the two subspaces obtained in the algorithm as the span of all eigenvectors above the largest gap and below this gap are close to $V,\bar{V}$.  This will follow using a version of Wedin's theorem for perturbations of matrices.
\item The total error accumulated by recursion remains below the target error $\eps$ for each column.
\end{enumerate}

\subsection{Maximum spacings of Gaussian polynomials}
Here we study the largest
gap between successive eigenvalues of the matrix $D^2 \log( \phi(u))$.
For a set of real numbers $x_1, \ldots, x_n$, define the maximum
gap function as:
\begin{align*}
  \maxgap{x_1, \ldots, x_n} = \max_{i \in [n]} \min_{j \in [n]: x_j
    \ge x_i} x_j - x_i 
\end{align*}
The $\textrm{maxgap}$ function is simply the largest gap between two
successive elements in sorted order.
\begin{theorem}\label{thm:maxgap}
  Let $\{ p_1(x), \ldots, p_n(x) \}$ be a set of $n$ quadratic
  polynomials of the form $p_i(x) = a_i x^2$ where $a_i >
  0$ for all $i$ and $\{z_1, \ldots, z_n\}$ be iid standard
  Gaussians. 
Then, with probability at least
  $1/(2000 \log^2 n)$,
  \begin{align*}
    \maxgap{ p_{1}(z_{1}), \ldots,  p_{n}(z_{n})} \ge
    \frac{1}{50} \min_i a_i.
  \end{align*}
  \end{theorem}
 We can simply repeat the experiment $O(\log^3 n)$ times to obtain a
  high probability guarantee. This type of $\textrm{maxgap}$ function
  has been somewhat studied in the mathematics literature -- there are
  a number of asymptotic results \cite{deheuvels1985limiting,
    deheuvels1984strong, deheuvels1986influence}. The rough intuition
  of these results is that asymptotically, the $\textrm{maxgap}$
  depends only on the tails of the random variables in question. Our
  work differs from these results in two very important ways --
  firstly, our results are quantitative (i.e., not simply in the limit
  of $n \to \infty$), and secondly, our result is true even if you
  pick the polynomial after fixing $n$. The latter, in particular,
  makes the problem quite a bit harder as now the family of random
  variables is no longer even uniform over $n$.
  \begin{proof}
  The first stage of the proof is to reduce the problem from the
  random model $\{ a_1 z_1^2, \ldots, a_n z_n^2 \}$ to sampling from a mixture model,
  which will more easily allow us to analyse the maximum gaps. To this
  end, let $f_i$ denote the distribution of $p_i(z_i)$, then consider
  the following uniform mixture model: 
\[
F(x) = \frac{1}{n} \sum_{i=1}^n f_i(x) = \frac{1}{n} \sum_{i=1}^n \Pr(a_iz_i^2 = x).
\] 
One
  can think of the simulation of a sample $x \sim F$ as a two-stage
  process. First, we pick an $i \in [n]$ uniformly at random (this
  gives a corresponding $a_i$), and then we pick $z \sim N(0,1)$
  independently. The product $a_i z^2$ then has distribution given by
  $F$.

  Suppose we pick $m = 10 n \log(n)$ samples as follows: first
  we pick $m$ times independently, uniformly at random from $[n]$
  (with replacement) to obtain the set $Y = \{y_1, \ldots,y_m
  \}$; then we pick $m$ independent standard Gaussian random variables
  $\{z_1, \ldots, z_m \}$, and finally compute component-wise products
  $\{ y_1 z_1^2, \ldots, y_m z_m^2 \}$.  Let $Y_1, \ldots, Y_n$ be a partition of $Y$ according to which $a_i$ is assigned to each $y_j$, i.e., $Y_i$ is the set of $y_j$'s for which $a_i$ was chosen. 

The following bounds follow from standard Chernoff-Hoeffding bounds.  For i.i.d.
    Bernoulli $\{0,1\}$ random variables with bias $p$:
    \begin{align}\label{eqn:chernoffagain}
      \prob{ \frac{1}{m} \sum_{i=1}^m X_i \ge (1+\delta) pm } \le
      \expn{ - \frac{\delta^2 pm}{3}}
    \end{align}
  \begin{claim}\label{claim:conc}
$\prob{\exists i \, : \,  |Y_i| = 0} \le \frac{1}{n^9} \quad \mbox{ and } \quad \prob{\exists i \, :\, |Y_i| > 40\log n} \le \frac{1}{n^2}$.
  \end{claim}
\hide{  
\begin{proof}
   The proof of (1) can be found in \cite{motwani2010randomized}. The
    proof of (2) follows from the trivial union bound over all $i$,
    and the following standard form of the Chernoff bound
    Thus, if we sum the indicator random variables $\chi_{y_i = a_1}$
    over all $i \in [m]$, then we take $pm = 10 \log(n)$ and $\delta =
    3$ which yields a probability bound of $1/n^3$. Now union bounding
    over all $a_j$ yields the desired answer.
  \end{proof}
}

  We now assume the above two events do not occur which happens with probability at least 
$2/n^2$. Next, we draw a
  subsample of size $n$ from the set $Y = \{ y_1 z_1^2, \ldots, y_m z_m^2
  \}$ to form the set $S$. To do so, we simply 
  pick a single representative uniformly at random from each $Y_i$. From the claim
  above, we know that each bucket has at least one element, and at
  most $ 40 \log(n)$ elements in it. The set $W$ is the set of values 
  $y_i z_i^2$ associated with the $n$ representatives we
  picked uniformly at random.  A simple observation is that $W$ is
  distributed exactly as the $\{ p_1(z_1), \ldots, p_n(z_n)\}$ in the
  statement of this theorem. In fact, each $a_i$ shows up exactly once
  in $W$, and is multiplied by $z^2$ for $z \sim N(0,1)$, and all these random
  variables are independent.

  Next, we condition on the event that $\arg\max_{i} y_i z_i^2$ and $\arg\min_{i} y_i
  z_i^2$ are picked in $W$. This occurs with probability at
  least $1/1600 \log(n)^2$ since no bucket is of size greater than $40
  \log(n)$ by Claim \ref{claim:conc}. With this assumption, it is clear that $\maxgap{W}
  \ge \maxgap{y_1, \ldots, y_m}$.  Thus, it suffices for us to analyse $\maxgap{y_1 z_1^2,
    \ldots, y_m z_m^2}$. Since the latter random variable is independent
  of which $y_j$ are picked for $W$, we have a reduction
  from our original random variable model $\{ a_1 z_1^2, \ldots, a_m
  z_m^2 \}$ to (slightly more) samples from a mixture model $F$.

  To lower bound the maximum gap, observe that the density $F(x)$ is
  continuous, has its maximum at $x=0$ and monotonically decays to 0
  as $x \to \infty$, since this is true for each of the component
  distributions $f_i$.  We will now pick thresholds $t_0, t_1$ such
  that $t_1-t_0$ is large, and there is good probability that no
  element of $W$ takes its value in the interval $[t_0, t_1]$ and at
  least one element of $W$ takes its value to the right of $t_1$. We
  pick $t_0$ s.t.
  \begin{align*}
    \prob{x \ge t_0} = \frac{1}{n \log(n)}
  \end{align*}
  and $t_1 = t_0 + 2\min_i a_i  $. Let $a_1 = \min_i a_i$. 
With these settings, we have,
  \begin{align*}
    \frac{ \prob{x \ge t_1}}{\prob{x \ge t_0}} = \frac{\sum_{i=1}^n\prob{a_i z_i^2 \ge t_1}}{\sum_{i=1}^n\prob{a_i z_i^2 \ge t_0}}\ge \min_i
    \frac{\prob{a_i z_i^2 \ge t_1}}{\prob{a_i z_i^2 \ge t_0}} =
    \frac{\prob{a_1 z_1^2 \ge t_1}}{\prob{a_1 z_1^2 \ge t_0}} = \frac{
      \prob{z_1 \ge \sqrt{t_1 / a_1}}}{\prob{z_1 \ge \sqrt{t_0 / a_1}}}
  \end{align*}
  where the second inequality follows by simply expanding the
  densities explicitly.
  
  Next, we will use the following standard Gaussian tail bound \cite{Feller}.
\begin{fact}
  For $z$ drawn from $N(0,1)$, and $t \in \R$,
\begin{align*}
  \left(\frac{1}{x}- \frac{1}{x^3}\right)e^{-x^2/2} \le \sqrt{2\pi}
  \Pr(z \ge x) \le \frac{1}{x} e^{-x^2/2}.
\end{align*}
\end{fact}
From this, we have that $t_0 \ge 2a_1 \ln n$ and 
\begin{align*}
    \frac{ \prob{x \ge t_1}}{\prob{x \ge t_0}} \ge \frac{\sqrt{\frac{a_1}{t_1}}\left(1-\frac{a_1}{t_1}\right)}{\sqrt{\frac{a_1}{t_0}}} \cdot \frac{\exp(-t_1 / 2a_1)
    }{\exp(-t_0 / 2a_1)} \ge \frac{1}{2e}.
\end{align*}
Thus, $\prob{x \ge t_1} \ge 1/(2en \log(n))$.

To conclude, observe that with small constant probability,
there are at most $100$ points $a_i z_i^2$ inside the interval $[t_0,
t_1]$, and there exists at least one point to the right of the
interval. Thus, there must exist one spacing which is at least $2a_1 /
100$. The failure probability is dominated by $1-1/(1600 \log^2 n)$ as the other terms
are of lower order, hence we can bound the failure probability by $1- 1/(2000
\log^2 n)$.
\end{proof}
For higher order monomials, a slight modification to this
argument yields:
\begin{theorem}
  Let $\{ p_1(x), \ldots, p_n(x) \}$ be a set of $n$ degree $d$
  polynomials of the form $p_i(x) = a_i x^d$ where $a_i > 0$ for all
  $i$ and $\{z_1, \ldots, z_n\}$ be iid standard Gaussians. Then, with
  probability at least $1/(2000 \log^2 n)$,
  \begin{align*}
    \maxgap{ p_{1}(z_{1}), \ldots,  p_{n}(z_{n})} \ge
    \frac{d}{50} \min_i a_i^{\frac{2}{d}} \log(n)^{\frac{1}{2} - \frac{1}{d}}.
  \end{align*}
\end{theorem}
\hide{
For the linear case
$p_i(x) = a_i x$ (effectively Gaussians with different variances), we
can give a high probability bound (with a slight decrease in the gap):
\begin{theorem}
  Let $\{ p_1(x), \ldots, p_n(x) \}$ be a set of $n$ degree
  polynomials of the form $p_i(x) = a_i x$ where $a_i > 0$ for all
  $i$ and $\{z_1, \ldots, z_n\}$ be iid standard Gaussians. Then, with
  probability at least $1 - 1/n^2$:
  \begin{align*}
    \maxgap{ p_{1}(z_{1}), \ldots,  p_{n}(z_{n})} \ge
    \frac{1}{50 \log(n)^2} \min_i a_i
  \end{align*}
\end{theorem}
The proof of the latter is straightforward in comparison to Theorem
\ref{thm:maxgap}: one can skip the coupling argument and simply pick $t_0 = \sqrt{2 \log(n /
  \log(n)^{3/2})}$ and $t_1 = \sqrt{2 \log(n / \log(n))}$ and proceed
with estimating the tail probabilities.
}
\subsection{Sample complexity and error analysis}
The analysis of the algorithm uses a version of the $\sin
\theta$ theorem of \cite{davis1970rotation}. Roughly
speaking, the largest eigenvalue gap controls the magnitude of the
error in each subspace $V_1$ and $V_2$ in the algorithm, each
recursive step subsequently accumulates error accordingly, and we have
to solve a nonlinear recurrence to bound the total error.
\inshort{In the full proof, we will use Wedin's theorem, Taylor's theorem with remainder,
and a concentration bound of Vershynin. Here we outline the main steps.}
\infull{We will use the following theorems in the proof.
The first is a  form of Wedin's Theorem from
\cite{stewart1990matrix}.
  \begin{theorem}[\cite{stewart1990matrix}]\label{thm:Wedin}
    Let $A, E \in \C^{m \times n}$ be complex matrices with $m \geq n$. Let $A$ have
singular value decomposition
\begin{align*}
A = [ U_1 U_2 U_3] \left( \begin{array}{cc} 
    \Sigma_1 & 0 \\
    0 & \Sigma_2 \\
    0 & 0 \\
  \end{array} \right) [ V_1^\ast V_2^\ast]
\end{align*}
and similarly for $\tilde{A}=A+E$ (with conformal decomposition using $\tilde{U}_1,
\tilde{\Sigma}_1$ etc).  Suppose there are numbers $\alpha, \beta >
0$ such that 
\[
\min \sigma( \tilde{\Sigma}_1) \ge \alpha + \beta \quad \mbox{ and } \quad
\max \sigma( \Sigma_2 ) \le \alpha.
\]
Then,
\begin{align*}
  \norm{ \sin (\Phi)}_2 , \norm{ \sin( \Theta)}_2 \le
 \frac{\norm{E}_2}{\beta}
\end{align*}
where $\Phi$ is the(diagonal) matrix of canonical angles between the ranges of $U_1$ and
$\tilde{U}_1$ and $\Theta$ denotes the matrix of canonical angles between the 
ranges of $U_2$ and $\tilde{U}_2$.
\end{theorem}
We will also need Taylor's theorem with remainder.
\begin{theorem}\label{thm:taylor}
  Let $f: \R \to \R$ be a $C^n$ continuous function over some interval
  $I$. Let $a,b \in I$, then
  \begin{align*}
    f(b) = \sum_{k=1}^{n-1} \frac{f^{(k)} (a)}{k!}(b-a)^k +
    \frac{f^{(n)}(\xi)}{n!} (b-a)^n,
  \end{align*}
  for some $\xi \in [a,b]$.
\end{theorem}

The following simple bounds will be used in estimating the sample complexity.
\begin{lemma}\label{lemma:sample1}
  Suppose that the random vector $x \in\R^n$ is drawn from an
  isotropic distribution $F$. Then for $1 \le j \le n$,
  \begin{align*}
\Var( x_j e^{iu^Tx} ) &\le 1, \quad \Var(x_j^2 e^{iu^Tx}) \le \E{x_j^4} \mbox{ and } 
\Var(x_i x_j e^{iu^Tx} )\le 1 \mbox{ for } i \neq j.
\end{align*}
\end{lemma}
\hide{
\begin{proof}
  \[
   \mathrm{Var} ( x_j
    e^{iu^Tx} ) =  \E{x_j^2} - \abs{\E{x_j e^{iu^Tx}}}^2
     \le 1.
  \]
The other parts are similar, with the last inequality using isotropy.
\end{proof}
}

The next is Theorem 1.2 from
 \cite{VershyninSingular}.
\begin{theorem}[\cite{VershyninSingular}]\label{thm:vershyninsingular}
   Consider a random vector $x \in \R^n$ with covariance $\Sigma$,
   such that $\norm{x} \le \sqrt{m}$ almost surely. Let $\eps \in
   (0,1)$ and $t\ge 1$, then with probability at least $1 -
   1/n^{t^2}$, if $N \ge C( t/ \eps)^2 \norm{ \Sigma}^{-1} m \log(n)$,
   then $\norm{\Sigma_N - \Sigma} \le \eps \norm{ \Sigma}$.
 \end{theorem}
}

\begin{proof}[of Theorem \ref{thm:main}]
If $A$ is not unitary to begin with, by Theorem \ref{thm:vershyninsingular}  a sample of size $O(n\log n/\eps_1^2)$ can be used to put $x$ in near-istropic position. After this transformation the matrix $\tilde{A}$ obtained is nearly unitary, i.e., $\|\tilde{A}- \bar{A}\|_2 \le \eps'$ where $\bar{A}$ is unitary. Henceforth, we assume that $\bar{A}=A$.
First, we prove that when we run 
  the algorithm and compute a set of eigenvalues, that there exists at least one large gap in
  the set $\{ \textrm{Re}(g_j( t_j)) \}$, the diagonal entries in the decomposition of $\textrm{Re}(\Sigma_u)$. 
We will do this for the unitary matrix $A$, and it will follow for the estimated matrix $\tilde{A}$, since their eigenvalues are within $\eps'$.

We recall that $g_j = \psi_j''$, and using the Taylor expansion of $\psi_j$ (\ref{psi}), we write each $g_j$
as follows.
\begin{align}\label{eqn:taylor}
  g_j(t_j) = -\sum_{l=2}^k \cum_l(s_j)\frac{(it_j)^{l-2}}{(l-2)!} - g^{(k+1)}(\xi)\frac{(it_j)^{k-1}}{(k-1)!}
\end{align}
where $ \xi \in [0,t_j]$ and $p_i$ is a polynomial of degree $(k-2)$.
Using $k=4$ and $j=1$,
  \begin{align*}
    g_1(t_1) = -1 -  \cum_3(s_1) (it_1) -  \cum_4(s_1) \frac{(it_1)^2}{2}
    + R_{1}(t_1) \frac{(it_1)^3}{3!}. 
  \end{align*} 
  When we take the real part of the matrix in Step 6 of the
  algorithm, we can discard the pure imaginary term arising from the
  first cumulant.  We must retain the error term as we do not
  know a priori whether the error derivative term has a complex
  component or not. Truncating after the second order terms,
  this gives a family of polynomials 
\[
p_j(t_j) = -1 + \cum_4(s_j)\frac{t_j^2}{2}.
\]
Since $\cum_4(s_j) \ge \Delta$ and $t_j$ is drawn from $N(0,\sigma^2)$,  we can now apply Theorem \ref{thm:maxgap} that shows that with
  probability $1/2000 \log^2 n$, 
\[
\maxgap{p_j(t_j)} \ge \frac{\Delta \sigma^2}{50}. 
\]
Thus with $8000 \log^3 n$ different random vectors $u$, with probability
  at least $1- (1/n^2)$ we will see a gap of at least this magnitude. 
\inshort{We then bound the remainder term using the version of Taylor's theorem with a remainder.}

\infull{
Next, we bound the remainder. Using Lemmas 4.9 and 10.1 from \cite{GVX14}, for $t_j \in [-1/4, 1/4]$, we have
\[
|R_j(t_j)| \le \frac{4! 2^4  \E{|s_j|^5}}{(3/4)^5}. 
\]
So the full remainder term with probability at least $1- (1/n^2)$ 
is at most
\begin{align*}
|R_j(t_j)\frac{(t_j)^3}{3!}| \le \frac{4^7}{3^5}\E{|s_j|^5}|t_j|^3 &\le \frac{4^7}{3^5}\E{|s_j|^5}\sigma^3 (4\log n)^{3/2}\le  \frac{1}{100} \Delta \sigma^2
\end{align*}
for 
\[
\sigma \le \frac{\Delta}{1000 (\log^{3/2} n)\E{|s_j|^5}}.
\]
}

Let $V$ and $V^{\perp}$ denote the sets of eigenvectors above and below the maximum gap
  respectively. We bound the error using \inshort{Wedin's theorem \cite{stewart1990matrix},}
\infull{Theorem \ref{thm:Wedin},} which bounds the canonical angles in terms of the gap. Suppose that in each iteration, we
  take enough samples so that the empirical version of $D^2 \psi(u)$ is within $\eps'$ of the true one. Then applying the theorem yields that for the subspaces spanned by $V$
  and $W = V^\perp$, that there exists a partition of the columns of $A$
  (which we may take, without loss of generality, to be ordered
  appropriately) such that:
  \begin{align*}
    \norm{ \sin (\Theta(V, \{A^1, \ldots, A^k \})) } \le
    \frac{50\eps'}{\Delta\sigma^2}.  
  \end{align*}
 For Recursive FPCA in the subspace $V$ of
  dimension $k$, we can write the matrix as:
  \begin{align*}
     D^2 \psi(u) & = (V^T[A^1, \ldots, A^k])\diag{ \lambda_1,
      \ldots, \lambda_k} (V^T[A^1, \ldots, A^k])^T \\
    & \qquad + (V^T [A^{k+1},
      \ldots, A^n]) \diag{\lambda_{k+1}, \ldots, \lambda_n} (V^T [A^{k+1},
      \ldots, A^n])^T
  \end{align*}  
The additional sampling error in this iteration (for the new $u$) is bounded by $\eps'$.
 By definition, we have that $\sin(\Theta) = V^T [A^{k+1},
    \ldots, A^n]$, thus the second term is upper bounded by
  $(50 \eps' / \Delta \sigma^2)^2$.  
For the first term, we can imagine making $V^TT(u,u)V$ isotropic, in which case we must account for the slight 
nonorthogonality of the columns of $V$ and $A^1, \ldots, A^k$. For this we have to multiply the error by 
\[
\|(V^T[A^1,\ldots,A^k])^{-1}\|_2 \le \|\cos(\Theta)^{-1}\| \le \|I+\sin(\Theta)^2\|
\]
where $\sin(\Theta)$ is the error accumulated so far, and we assume that $\|\sin(\Theta)\| \le 1/2$. 

 We can write the recurrence for the overall error $E_i$ at a
  recursive call at depth $i$, then:
  \begin{align*}
    E_{i} \le (1+E_{i-1}^2)\left(\eps' + \left( \frac{E_{i-1}}{ c \Delta\sigma^2}\right)^2\right).
  \end{align*}
 We apply Claim \ref{claim:recurrence} (in the appendix) with $a = \eps'$ and $ b = \Delta\sigma^2/50$. 
 To satisfy the condition it suffices to have
$\eps' \le (\Delta\sigma^2/50)^2/8$, which we will achieve by setting 
$\eps' = \eps\Delta\sigma^2/100$
since $\eps < \Delta^3/(10^8\M_5^2\log^3 n)$.   In the terminal nodes of the recurrence, the error gets blown up to at most $2\eps'$  and this implies a final error between the output vectors and the columns of $A$ of at most $2\eps'/(\Delta\sigma^2/50) \le \eps$.

 For the sample complexity of a single eigendecomposition, we have to take enough samples so
 that for $8000 \log^3 n$ different instantiations of $D^2 \psi(u)$, the spectral norm error is within $\eps'$
 with high probability. 
It suffices to estimate three matrix-valued random
 variables $\E{ xx^T \exp( iu^Tx)}$, $\E{x \exp(iu^Tx)}$ and
 $\E{\exp(iu^Tx)}$. The latter two are easy to estimate using $O(n)$
 samples \infull{by applying Lemma \ref{lemma:sample1}}. Thus, it suffices for
 us to show that we can estimate the second order term  $\E{ xx^T
   \exp( iu^Tx)}$ using only a nearly linear number of samples. 
   \infull{
  We rewrite this term as four easy-to-estimate parts:
\begin{align*}
  \E{xx^T \exp(iu^Tx)} 
  & = \E{xx^T \mathbbm{1}_{\cos(u^Tx) \ge 0} \cos(u^Tx)} - \E{xx^T
    \mathbbm{1}_{\cos(u^Tx) < 0} \abs{\cos(u^Tx)}} \\ 
  & \qquad + i \E{xx^T \mathbbm{1}_{\sin(u^Tx) \ge 0} \sin(u^Tx)} - i
  \E{xx^T \mathbbm{1}_{\sin(u^Tx) < 0} \abs{\sin(u^Tx)}}
\end{align*}
We estimate the four terms using independent samples to within error $\eps/4$ in spectral norm. Consider, for
example, the first term $xx^T \mathbbm{1}_{\cos{u^Tx} \ge 0}
\cos(u^Tx)$, then we can define the random vector 
\[
y = x
\mathbbm{1}_{\cos(u^Tx) \ge 0} \sqrt{\cos(u^Tx)}, 
\]
so that the first term is $yy^T$.
In particular, observe that $0 \le \E{(u^Ty)^2} \le \E{(u^Tx)^2} \le
1$ for all unit vectors $u$. Thus, we must have that the eigenvalues
of $\E{yy^T}$ are all bounded by 1. Note also, that $\norm{y} \le
\sqrt{m}$ if this is in fact the case for $x$ as well. Now, we apply
Theorem 1.2 from \cite{VershyninSingular} to $y$: by hypothesis, we
can take $m = K^2 n$ and $t = 2$. 
}
Let $N$ be the number of samples needed to get a failure probability smaller than $1/n^2$.

Next, we argue that we can re-use the $N$ samples from the initial phase for
the entire algorithm (without re-sampling), and apply the union
bound to get a failure probability bounded by $1/n^2$, thereby giving us a
high probability of success for the entire algorithm.
To find a good  vector $u$, we form the second derivative matrices for a set of $u$'s, explicitly compute the eigenvalue gaps, find the largest gap for each matrix, and pick the $u$ whose matrix has the largest gap. Finding a good $u$ here depends only on the randomness of $u$.
We note that each level of decomposition we estimate is of the form 
$\E{(P^Tx)(P^Tx)^T\expn{i u^Tx}} = P^T\E{xx^T\diag{\expn{i u^Tx}}}P$, with the expectation independent 
of the current projection. We will do this 
for at most $O(n\log^3 n)$ different vectors $u$, which are random and independent of the sample.
To achieve overall error $\eps$, the sample complexity is thus
\[
O\left(K^2 n \frac{\log n}{(\eps')^2}\right) = O\left(n \cdot \frac{K^2 \M^4 \log^7 n}{\Delta^6\eps^2} \right).
\]
\end{proof}

\section{Conclusion}

Our work was motivated by experiments on Fourier PCA and tensor-based
methods, which appeared to need a rather large number of samples even
for modest values of the dimension $n$. In contrast, the recursive algorithm 
presented here scales smoothly with the dimension, and is available as
MATLAB code \citep{Xiao14}. 

Analyzing the gaps of a family of polynomials over
Gaussians is an interesting problem on its own. One surprise here is
that even for degree $3$, the polynomial $p(x) = x(x-a)(x+a)$ where $a
= \sqrt{2 \log(n)}$ evaluated at $n$ random points from $N(0,1)$ has
maximum gap only $O(1/n^{0.6})$, no longer polylogarithmic as in the case of 
degree $1$ or $2$.

\bibliography{ICA_bibliography,stat_algs}

\section{Appendix}
  \begin{claim}\label{claim:recurrence}
  Let $a, b \in [0,1]$, $a \le b^2/8$, and define the recurrence:
\[
y_0 = 0, \quad 
y_{i+1} = (1+y_i^2)\left(a + \left(\frac{y_i}{b}\right)^2\right).
\] 
Then $y_i \le 2a$ for all  $i$.
\end{claim}
\begin{proof}
  We proceed via induction. Clearly this is true for $i=0$. Now
  suppose it is true for up to some $i$. Then
\[    y_{i+1} = (1+4a^2)(a + (2a / b)^2) \le a\left(1+ 4a^2 + 4\frac{a}{ b^2} + 16\frac{a^3}{b^2}\right) \le 2a
\] 
since $a \le b^2/8 \le 1/8$.
\end{proof}

\noindent
{\bf Acknowlegements.}  We are grateful to the anonymous referees for their suggestions. This work was supported in part by NSF awards  CCF-1217793 and EAGER-1415498.

\end{document}